\documentclass[10pt,conference,letterpaper,table]{IEEEtran}
\IEEEoverridecommandlockouts
\usepackage{cite}
\usepackage{amsmath,amssymb,amsfonts}
\usepackage{amsthm,bbm,breqn}
\usepackage{graphicx}
\usepackage{textcomp}
\usepackage{hyperref}
\usepackage[nameinlink]{cleveref}
\usepackage{enumitem}
\usepackage{xcolor}
\usepackage{float}
\usepackage{forest}
\newtheorem*{definition*}{Definition}

\newtheorem{claim}{Claim}
\newtheorem{lemma}{Lemma}
\newtheorem{theorem}{Theorem}
\newtheorem{corollary}{Corollary}

\def\*#1{\mathbf{#1}}

\newlist{procsteps}{enumerate}{1}
\setlist[procsteps,1]{label=\arabic*., ref=\arabic*}

\crefname{procstepsi}{Procedure}{Procedures}

\begin{document}

\title{Minimum-hop Constellation Design for Low Earth Orbit Satellite Networks
}

\author{\IEEEauthorblockN{Chirag Rao\textsuperscript{\textdagger *} and Eytan Modiano\textsuperscript{\textdagger}}
\IEEEauthorblockA{\textit{ \textsuperscript{\textdagger}Laboratory for Information \& Decision Systems, MIT}\\
\textit{\textsuperscript{*}DEVCOM Army Research Laboratory} \\ \textit{crao@mit.edu, modiano@mit.edu}}\thanks{We would like to thank SES S.A. for supporting this research.}
}

\maketitle
\begin{abstract}
We consider a Low Earth Orbit (LEO) satellite network with each satellite capable of establishing inter-satellite link (ISL) connections for satellite-to-satellite communication. Since ISLs can be reoriented to change the topology, we optimize the topology to minimize the average shortest path length (ASPL). We characterize the optimal ASPL ISL topology in two families of topologies, 1) vertex-symmetric in which the ISL connections at a satellite node represent a motif that is repeated at all other satellite nodes, and 2) general regular topologies in which no such repeating pattern need exist. We establish ASPL lower bounds for both scenarios and show constructions for which they are achievable assuming each satellite makes 3 or 4 ISL connections. For the symmetric case, we show that the mesh grid is suboptimal in both ASPL and diameter. Additionally, we show there are constructions that maintain intra-orbital ISL connections while still achieving near-optimal ASPL performance. For the general case we show it is possible to construct networks with ASPL close to the general lower bound when the network is sufficiently dense. Simulation results show that for both scenarios, one can find topologies that are very close to the lower bounds as the network size scales.
\end{abstract}

\section{Introduction}
 Network topology design is critical for ensuring optimal operation of satellite networks. Network operators have to design topologies to meet the stringent Quality of Service demands required for user traffic. For large satellite networks with thousands of satellite nodes, such as the proposed constellations for Starlink and Kuiper, topology design becomes even more critical for ensuring network operation can be done with as little overhead as possible, since deploying satellite systems is financially costly~\cite{werner2001topological}.  
 However, the introduction of high-bandwidth optical inter-satellite links (ISLs) for the space segment of the satellite network offers some topology design flexibility: since ISLs can be reoriented to establish a link between any two satellites within link range, the space-segment network topology can be reconfigured after deployment of the constellation. 
 Given the freedom to choose from a large set of possible topologies, choosing the best network topology is desirable. 
 
 Latency is a metric of paramount importance in satellite networks for users that need near real-time performance, such as those using the network for military applications or for commercial operations with strict time constraints such as high-frequency trading~\cite{handley2018delay}. 
 Many strategies exist for reducing or limiting network latency, such as routing for mitigating queueing backlog at forwarding satellite nodes~\cite{cao2013per,handley2018delay}, and resource management strategies that guarantee particular quality of service criteria, including latency~\cite{zanzi2020laco}. 
 Even so, the topology strongly influences how routing schemes behave and how much overhead they incur; 
 it also determines how much network capacity is available for resource management. 
The most direct effect a topology has on latency is the distance a packet travels.

 A common metric for measuring distance is the average shortest path length (ASPL) of the network. 
 The ASPL captures the distance a packet must travel on average through the network to reach its destination assuming a shortest-path routing scheme like Dijkstra's Algorithm is used to choose the shortest network path. 
 ASPL serves as a proxy for measuring latency, since each hop along the path in the network incurs delay and contributes to the latency. 
Other metrics exist for measuring latency, such as geographical distance and direct end-to-end latency probing, and have been well-studied in the literature~\cite{wang2008constellation,chen2021analysis}.
 ASPL and its relation to average geographical path distance was studied in~\cite{chen2022leo}, which demonstrated there exists parameter settings for which the two metrics are equivalent. 
 The geographical distance between satellite nodes is subject to change since LEO orbits are not geostationary. 
 Therefore, a geographical distance-optimal topology can quickly become suboptimal for that metric, whereas the optimal ASPL topology would continue to be optimal as long as the set of visible satellites remains unchanged.
 In~\cite{wang2008constellation}, delay probing is used to infer end-to-end latency caused by various sources of delay. While delay probing is appropriate for measuring latency in real-time, latency for topology design must be measured with respect to a larger timescale: probing packets traverse a satellite network on the order of milliseconds whereas topologies can only be reconfigured via ISLs on the order of minutes~\cite{bhattacharjee2023laser}. 
 We therefore study the optimal design of satellite-to-satellite network topology with respect to the ASPL.

There are several models for representing ISL networks in the literature, capturing aspects of the constellation that affect topology design, such as orbital parameters (altitude, degree of ascension), communication power budget, antenna design, and ISL setup time~\cite{werner2001topological,henderson2000network,soret2019inter}. Many standard constellation designs exist such as the Walker Delta or Polar constellations, which can influence the choice of ISL topology~\cite{royster2023network}. In this paper, we assume the satellites' positions are fixed, which is a common model in the literature~\cite{wang2007topological}. We study ASPL-optimal topology design assuming each satellite node has 1 unit of traffic for every other node, and each satellite makes exactly $\Delta$ ISL connections with other satellites.  
We represent the topology as a graph --- with satellites as vertices and ISL connections as edges --- and consider two broad topology design cases: vertex-symmetric topology design, and general regular topology design. 
Vertex-symmetric topologies can be expressed as a repeating motif in the network where connections between vertices follow the same pattern, such as the mesh grid~\cite{bhattacherjee_network_2019}, whereas general regular topologies need not exhibit such structure. There are distinct advantages to vertex-symmetric topologies, since the regular structure provides many alternative shortest paths between vertices and provides ease of management. 
But, as will be shown, imposing vertex symmetry on a network of size $N$ can limit the ASPL performance to scale $\Theta(\sqrt{N})$, and improves to $\Theta(\log N)$ without the symmetry condition. 
For each topology design case, we analyze the optimal ASPL topology assuming that the degree is $3$ or $4$, the latter a typical model in the literature and in practice~\cite{sun_capacity_2003}, and the former a potentially more cost-effective alternative.

Our primary contributions in this paper are the following.  We develop lower bounds for the vertex-symmetric case depending on whether the number of ISLs per satellite is even or odd valued. Next, we show that the mesh grid topology is suboptimal in both ASPL and graph diameter.
We show the vertex-symmetric lower bounds are achievable for degree 3 and 4. In particular, we show that when the degree is 4 there exist topologies that meet the lower bound and are comparable in link range to that of the mesh grid.
For the general regular case we develop a procedure that with high probability finds low-ASPL topologies that scale according to the general regular topology ASPL lower bound.

The remainder of the paper is organized as follows. In~\Cref{sec:model} we mathematically define the network model. In~\Cref{sec:LB} we develop lower bounds for the vertex-symmetric design case when the degree is 3 and 4. 
In~\Cref{sec:results}, we search for ASPL-optimal topologies for both design cases.
In~\Cref{sec:conclusion} we provide concluding remarks.

\section{Network Model}\label{sec:model}
\begin{figure}
    \centering
\includegraphics[width=0.48\textwidth]{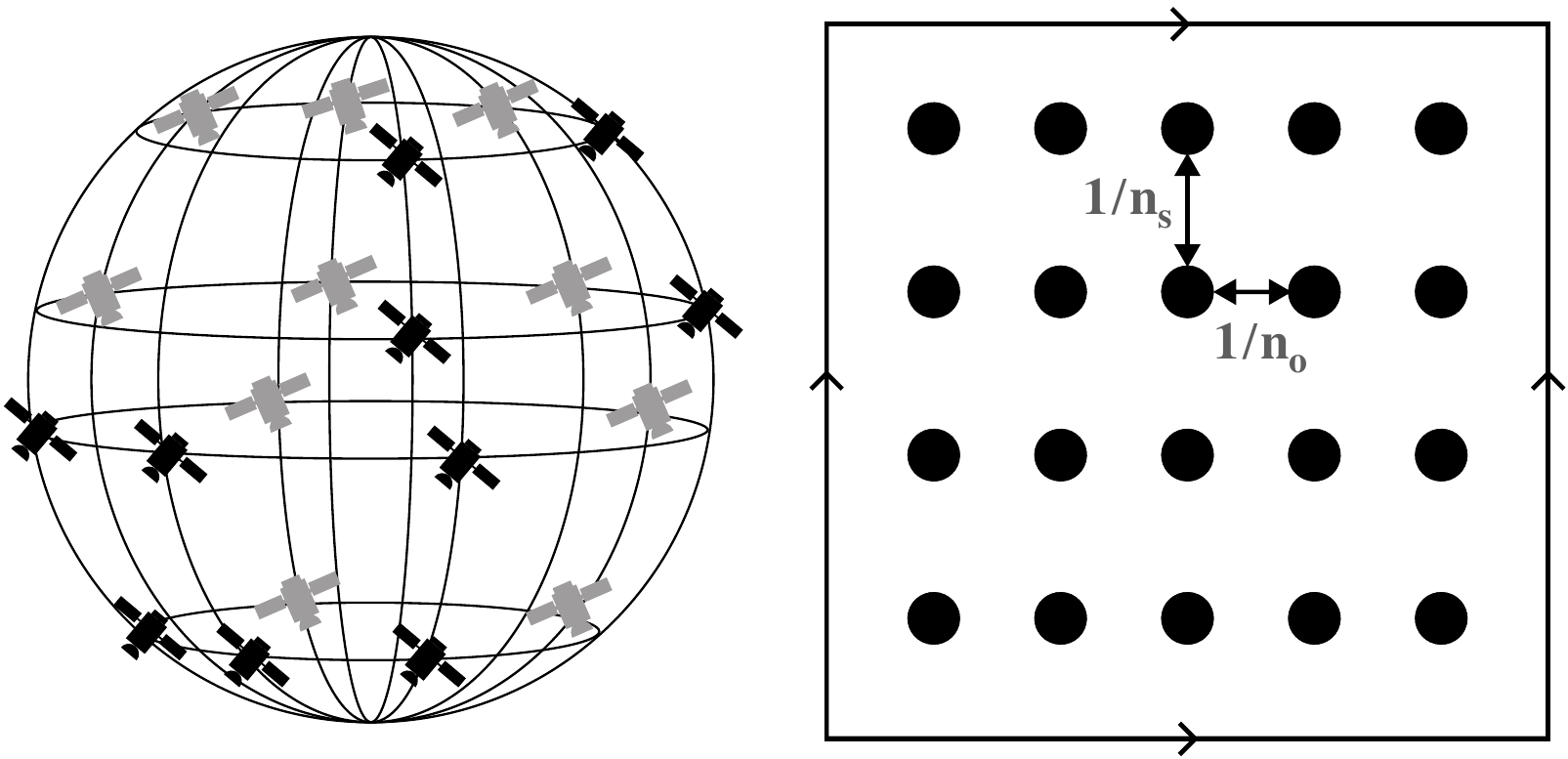}
    \caption{The constellation on the left is represented as a lattice lying in a flat square torus with unit area on the right.}
    \label{fig:earth_to_torus}
\end{figure}
In this work, we use the terms satellite, vertex, and node interchangeably. We model the network as a static $n_s\times n_o$ constellation where the number of satellites is $N=n_sn_o$.  The satellites are arranged into $n_s$ satellites per orbital plane (rows) and $n_o$  orbital planes (columns) in a two-dimensional lattice formation. Each node $\*x$ is uniquely identified by its row and column position in the lattice $\*x\in\{0,\hdots,n_s-1\}\times\{0,\hdots,n_o-1\}$. The lattice formation is embedded into a unit area square torus to represent the wraparound nature of satellite orbits (see~\Cref{fig:earth_to_torus}). To first order, one can use this setup to model spatial relations in a satellite constellation. Clearly, this is an abstraction and does not capture the geometry of a satellite constellation with full fidelity. However, it is useful for gaining insight into topology design~\cite{modiano1996efficient,sun_capacity_2003}. 

A network topology $G$ can be represented as an undirected graph $G=(V,E)$, with vertex set $V$ representing the satellite nodes and the edge set $E$ representing the ISL connections established between satellite nodes. Each satellite must establish $\Delta$ ISL connections. Thus, the graph $G$ is $\Delta$-regular, with each vertex having degree $\Delta$. 
We define the operator ``$\oplus$'' when applied to a vertex $\*v\in V$ and 2-dimensional integer vector $\*e\in\mathbb{Z}^2$ represent the following operation:
$$\*v\oplus \*e\triangleq \begin{bmatrix}
    v[0]+e[0]\pmod{ n_s} \\ v[1]+e[1]\pmod {n_o}
\end{bmatrix}\,,$$
For example, in a $7\times 5$ constellation with $\*v=[1\ 2]^\top$ and $\*e=[4\ 6]^\top$, $\*v\oplus \*e=[5\ 3]^\top$. Thus the set of vertices $V$ is closed under ``$\oplus$". We similarly define ``$\ominus$'' to be $\*v\ominus \*e=\*v\oplus -\*e$.

 A graph with even-valued degree $\Delta$ is defined as vertex-symmetric if for any node $\*v\in V$ and a set $\{\*e_1,\*e_2,\hdots,\*e_{\Delta/2}\}\subset \mathbb{Z}^2$ its adjacent neighbors are the set of nodes $\left\{\*v\oplus \*e_1,\hdots,\*v\oplus \*e_{\Delta/2}, \*v\ominus \*e_1,\hdots, \*v\ominus \*e_{\Delta/2}\right\}$.
 The vectors $\*e_1,\hdots,\*e_{\Delta/2}$ are called \textit{jumps}, a collection of which is a \textit{jump set}.
 An example of a vertex-symmetric topology is the mesh grid, defined by the jump set $\{[1\  0]^\top, [0\ 1]^\top\}$ and illustrated in~\Cref{fig:meshgrid}.
\begin{figure}
    \centering
    \includegraphics[width=0.45\textwidth]{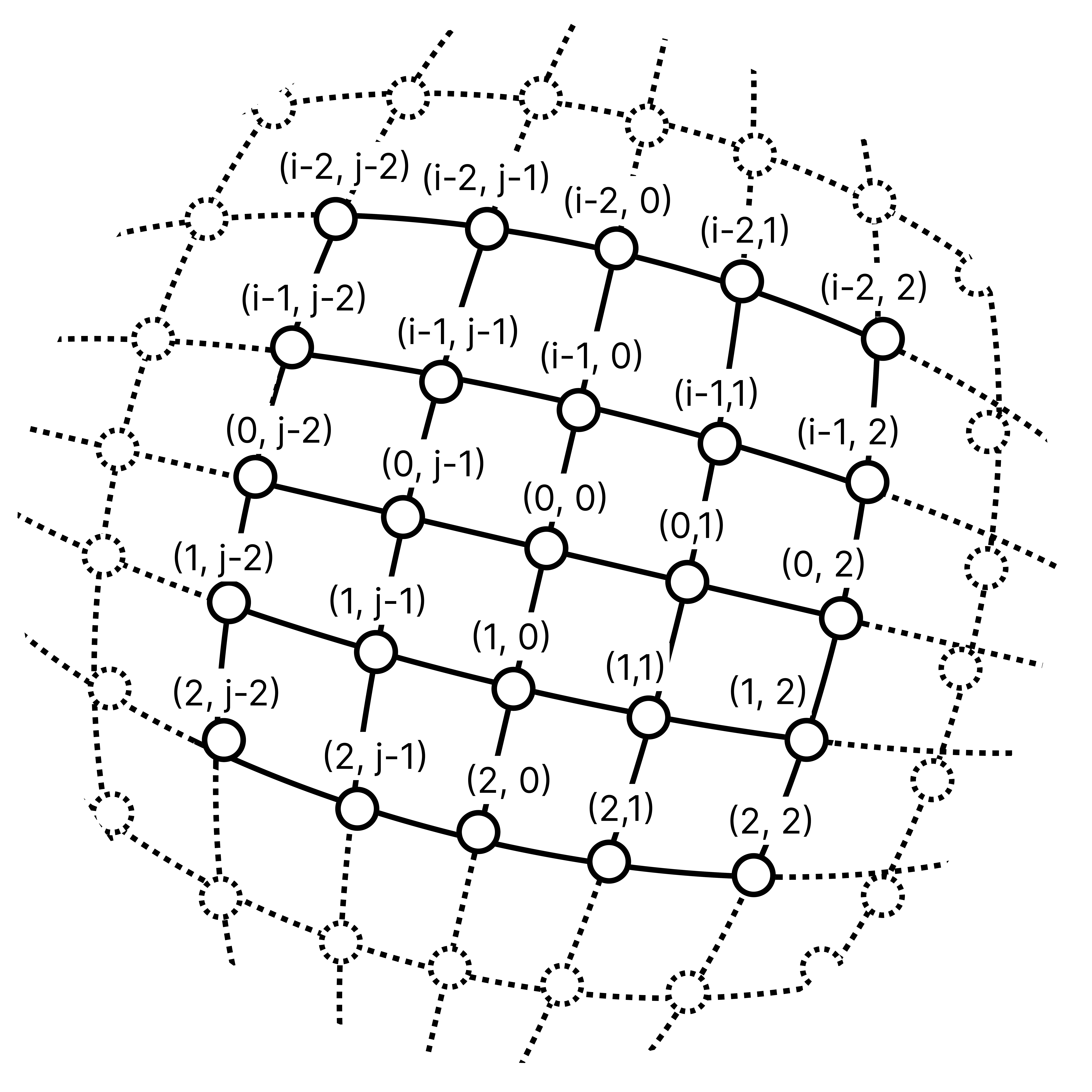}
    \caption{The mesh grid topology for $n_s=i$ and $n_o=j$.}
    \label{fig:meshgrid}
\end{figure}
We now enumerate the number of unique nodes that are reachable with shortest path length $h$ with an example illustrated in~\Cref{fig:evenDegTree}. 
\begin{figure}
    \centering
    \includegraphics[width=0.35\textwidth]{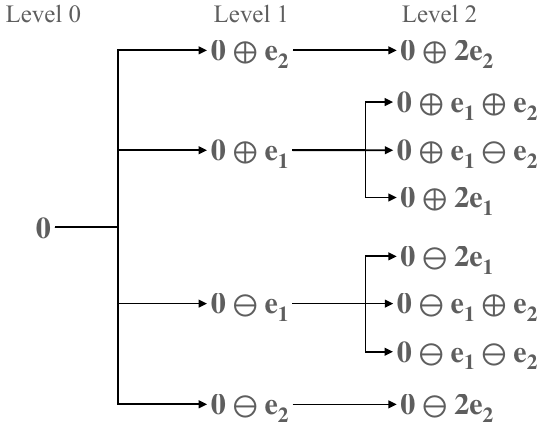}
    \caption{Vertex-symmetric shortest path length coverage for $\Delta=4$}
    \label{fig:evenDegTree}
\end{figure}
The tree in~\Cref{fig:evenDegTree} is for a vertex symmetric network with $\Delta=4$, and it can be formed by breadth-first search through the graph.
The jump set $\{\*e_1,\*e_2\}$ determines which hops can be taken from a node. The corresponding hop set is $\{\*e_1,-\*e_1,\*e_1,-\*e_2\}$.
Level $i$ of the tree shows the set of nodes that are reached with a shortest path length of $i$.
Without loss of generality, the graph traversal begins at node $\*0$ at Level $0$.
The nodes at Level $1$ are reached by traversing along the 4 distinct hops.
To reach Level $2$ nodes, paths are composed of two hops.
Note that 
\begin{align*}
    \*0\oplus \*e_1\oplus \*e_2=\*0\oplus \*e_2 \oplus \*e_1\,,
\end{align*}
so the commutativity of ``$\oplus$" implies the order in which hops are taken does not uniquely determine the destination. 
Additionally, no path in~\Cref{fig:evenDegTree} includes $\*e_i$ and $-\*e_i$ simultaneously, because the net effect would be traversing a path that is not the shortest path.

For a regular graph with odd degree, the number of vertices must be even-valued, since $2|E|=\Delta|V|$.  For odd-valued $\Delta$, the same notion of vertex symmetry as defined for even-degree cannot hold, since for odd $\Delta$ the jump-set could not have $\Delta/2$ elements. 
Therefore, we introduce a different definition of symmetry for odd $\Delta$. We partition the set of satellites into two equally-sized sets, $V_R$ and $V_B$. The set is partitioned in a spatially symmetric manner. Specifically, if
$$\mathbbm{1}^R_{\*v}=\begin{cases}
    1,\quad \*v\in V_R\\
    0,\quad \*v\in V_B
\end{cases}\,,$$
then for any $\*u,\*v\in V_R$ and $\*u\neq \*v$, and a hop $\*e$,  $\mathbbm{1}^R_{\*u\oplus \*e}=\mathbbm{1}^R_{\*v\oplus \*e}$. Likewise, for $\mathbbm{1}^B_{\*v}=1-\mathbbm{1}^R_{\*v}$, any $\*u',\*v'\in V_B$ and $\*u'\neq\*v'$ and hop $-\*e$, we have $\mathbbm{1}^B_{\*u'\ominus \*e}=\mathbbm{1}^B_{\*v'\ominus \*e}$.
An example of a partition that meets this condition is shown in~\Cref{fig:example_d3}. 
Given such a partitioning, we define vertex symmetry for odd degree. 
Define a jump set $\{\*e_1,\*e_2,\hdots,\*e_\Delta\}$. 
A graph is symmetric if for any node $\*v\in V_B$ its adjacent neighbors are the set of nodes $\{\*v\oplus \*e_1,\*v\oplus \*e_2,\hdots, \*v\oplus \*e_\Delta\}$. This condition is sufficient for vertex symmetry. An alternative condition for vertex symmetry is if for any node $\*u\in V_R$ its adjacent neighbors are the set of nodes $\{\*u\ominus \*e_1,\*u\ominus \*e_2,\hdots, \*u\ominus \*e_\Delta\}$. 
This node partitioning induces a partitioning of hops such that exactly $\Delta$ types of hops are available at each node, which guarantees $\Delta$-regularity. 
An example vertex-symmetric topology for $\Delta=3$ is the honeycomb mesh in~\Cref{fig:example_d3}, with jump set $\{[1\ 0]^\top,[-1\ 0]^\top,[0\ 1]^\top\}$. 
3-regular symmetric topologies have been studied in the context of interconnection networks for computer architectures~\cite{stojmenovic1997honeycomb}, but have not been analyzed in the satellite networks literature to the best of our knowledge.
\begin{figure}
    \centering
    \includegraphics[width=0.45\textwidth]{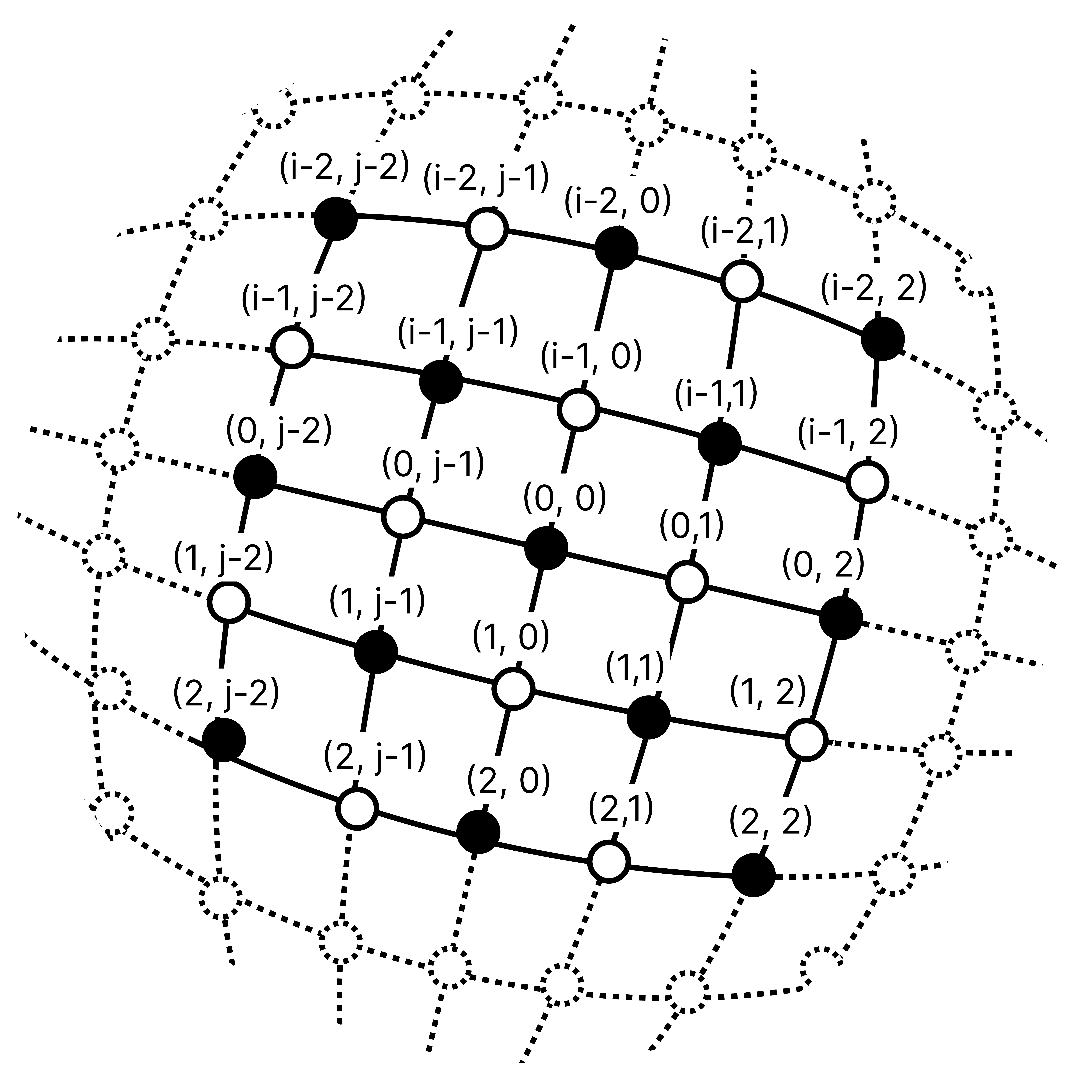}
    \caption{Example vertex-symmetric topology with $\Delta=3$: the honeycomb mesh, $n_s=i$, and $n_o=j$. The shaded nodes comprise one partition. The unshaded nodes comprise the other equally-sized partition.}
    \label{fig:example_d3}
\end{figure}

 Finally, a general regular topology is any simple $\Delta$-regular graph that does not contain any self-loops or double edges.
The graph need not exhibit any vertex or spatial symmetry, but each vertex has degree $\Delta$ to ensure each satellite uses its $\Delta$ ISLs to establish connections.
\section{Lower Bounds on ASPL}\label{sec:LB}
\subsection{Degree 4 Vertex-Symmetric Topologies}
We now derive a lower bound on ASPL for vertex-symmetric networks with degree 4 and 3. 
We state the following lemma that will be helpful for deriving the lower bound for $\Delta=4$. The lemma comes from adapting Theorem 1 from~\cite{boesch_reliable_1985}
\begin{lemma}\label{lem:boesch}
    In a $\Delta$-regular vertex-symmetric graph where $\Delta$ is even-valued, the maximum number of nodes that can be reached by a root vertex with a path length of up to $\ell$ hops cannot exceed $Y$, where 
    \begin{align}
    Y=\sum_{j=1}^t\binom{\Delta/2}{j}\binom{\ell-1}{j-1}2^j\,, \label{eq:boesch}
    \end{align}
    and $t=\min(\frac{\Delta}{2},\ell)$.
\end{lemma}
Substituting $\Delta=4$ into~\Cref{eq:boesch}, we see that for all non-negative $\ell$, the maximum number of unique nodes reached with shortest path length $\ell$ is $4\ell$. Applying this upper bound iteratively, we get that the maximum number of distinct nodes that can be reached in up to $h$ hops is
    $1+\sum_{\ell=1}^h4\ell=2h^2 +2h +1\,.$
This quantity grows quadratically in $h$. Using this fact along with~\Cref{lem:boesch} leads to a lower bound on average shortest hop for a network of size $N$ and degree $4$.
\begin{claim}\label{cl:vs4lb}
    Define $A(G)$ as the ASPL of graph $G$. The ASPL for any degree 4 vertex-symmetric graph $G$ is lower bounded by 
    \begin{dmath}\label{eq:vs4lb}
        A(G)\geq\frac{1}{N-1}\left[(k+1)(N-2k^2-2k-1)+\sum_{i=1}^k 4i^2\right]
    \end{dmath}
\end{claim}
where $k=\lfloor \frac{-1+\sqrt{2N-1}}{2}\rfloor$\,.
\begin{proof}
    We prove~\Cref{cl:vs4lb} by showing that minimizing ASPL is equivalent to successively maximizing the number of neighbors in each level, starting at the closest level. We cast the ASPL problem as the relaxation of a mixed integer linear program:
    \begin{align}
        \min_{x_1,\hdots,x_{N-1}}\quad &\frac{1}{N-1}\sum_{i=1}^{N -1} ix_i\label{eq:packing_objective}\\
        \text{s.t.}\quad &x_i\leq 4i,\, i=1,\hdots,N-1 \label{eq:capacity_constraint}\\
        &\sum_{i=1}^{N-1}x_i=N-1 \label{eq:ell_constraint}\\
        &x_i\geq 0,\, i=1,\hdots,N-1
    \end{align}
where each $x_i$ denotes the number of nodes with shortest path length $i$,~\Cref{eq:capacity_constraint} is the direct application of the upper bound from~\Cref{lem:boesch}, and~\Cref{eq:ell_constraint} ensures all $N-1$ nodes in the network are spanned. By inspection, we see that the optimal approach is to proceed in ascending order of $i$ and maximize the value of $x_i$. If the constraint in~\Cref{eq:capacity_constraint} is active, then one proceeds to maximize the value of $x_{i+1}$. If the constraint in~\Cref{eq:ell_constraint} is met upon setting the value of $x_{k+1}$ for some $k$, then all remaining decision variables $\{x_i\}_{i>k+1}$ are set to $0$. The optimal solution as a result of this procedure is exactly a greedy packing of as many nodes into each successive level as possible. The solution is unique, since any other optimal solution requires reducing the value of $x_i$ for some $i\leq k$ and increasing the value of $x_j$ for $j>k$ so that the constraint in~\Cref{eq:ell_constraint} can be met. Since $j>i$, the objective value would strictly increase, which is a contradiction.
To complete the proof we find $k$, which is omitted here for brevity, and can be found in the appendix.
\end{proof}
As a corollary to~\Cref{cl:vs4lb}, we have the following.
\begin{corollary}~\label{cor:lim_vs4lb}
    The lower bound in~\Cref{cl:vs4lb} converges to $\frac{\sqrt{2N}}{3}\approx 0.47\sqrt{N}\,$ as $N\to\infty\,.$
\end{corollary}
A simple end behavior analysis of~\Cref{eq:vs4lb} leads to~\Cref{cor:lim_vs4lb}.
The best possible ASPL performance using the standard mesh grid topology is $\sim 0.5\sqrt{N}$, assuming a square $n_s\times n_s$ constellation~\cite{stojmenovic1997honeycomb}. Although the ASPL lower bound appears close to the mesh grid ASPL performance, we will show in~\Cref{sec:results} that the mesh grid performs strictly worse than the ASPL lower bound for any sufficiently large constellation.

The value $k$ found in the proof of~\Cref{cl:vs4lb} is the largest level that is packed to capacity. Therefore the diameter is at least $k$, and it is straightforward to find that the diameter lower bound is
\begin{align}\label{eq:diam_lb}
    D_{LB}=\bigg\lceil \frac{-1+\sqrt{2N-1}}{2} \bigg\rceil\,.
\end{align}
\subsection{Degree 3 Vertex-Symmetric Topologies}
Next, we develop a lower bound for the case $\Delta=3$. Similar to the line of reasoning for $\Delta=4$, we first develop an upper bound on the number of distinct nodes reachable in exactly $\ell$ hops. The proof of the following is omitted for brevity.
\begin{claim}
    In a $\Delta$-regular vertex-symmetric graph where $\Delta$ is odd-valued, the maximum number of nodes that can be reached by a root vertex with a path length of $\ell$ hops cannot exceed $Y$, where
    \begin{align}\label{eq:odd_max_packing}
        Y=\sum_{i=1}^t \binom{\Delta}{i}\binom{\lceil \ell/2\rceil -1}{i-1}\left[\sum_{j=1}^{w_i} \binom{\Delta-i}{j}\binom{\lfloor \ell/2\rfloor-1}{j-1}\right]\, ,
    \end{align}
    where $t=\min(\lceil \ell/2\rceil ,\Delta)$ and $w_i=\min(\lfloor \ell/2\rfloor,\Delta-i)$
\end{claim}
As a corollary to the previous claim, we can determine the maximum number of nodes that can be reached in exactly $\ell$ hops when $\Delta=3$:
\begin{corollary}\label{cl:deg3_max_packing}
    For $\Delta=3$, the maximum number of nodes that can be reached by a vertex with a path length of exactly $\ell$ hops cannot exceed $Y$, where
    $Y=3\ell\,.$
\end{corollary}
\begin{proof}
    For $\lfloor \ell/2\rfloor<3$ one can directly verify by plugging in the value of $\ell$ and $\Delta=3$ into~\Cref{eq:odd_max_packing}. For $\lfloor \ell/2\rfloor\geq3$ and $\Delta=3$, \Cref{eq:odd_max_packing} evaluates to $Y=3\lceil \ell/2 \rceil + 3\lfloor \ell/2 \rfloor$. If $\ell$ is odd, then $Y=3\frac{\ell+1}{2}+3\frac{\ell-1}{2}=3\ell$. If $\ell$ is even, $Y$ again reduces to $3\ell$.
\end{proof}
Recall the maximum number of nodes that can be reached in exactly $\ell$ hops for the vertex-symmetric $\Delta=4$ case was $4\ell$; now with $\Delta=3$, the maximum is $3\ell$. 
Using~\Cref{cl:deg3_max_packing}, we find the maximum number of unique nodes that can be reached in up to $h$ hops with degree 3, which is
\begin{align}\label{eq:deg3_max_coverage}
    1+\sum_{\ell=1}^h3\ell=\frac{3}{2}h^2+\frac{3}{2}h+1\,.
\end{align}
 Given~\Cref{eq:deg3_max_coverage}, the lower bound on ASPL for $\Delta=3$ can be found with the same analysis presented in~\Cref{cl:vs4lb} to arrive at the following:
\begin{claim}\label{cl:vs3lb}
     For $\Delta=3$, the lower bound on ASPL for vertex-symmetric graph $G$ is 
    $$A(G)\geq\frac{1}{N-1}\left[(k+1)\left(N-\frac{3}{2}k^2-\frac{3}{2}k-1\right)+\sum_{i=1}^k 3i^2\right]\,,$$
    where $k=\bigg\lfloor \frac{-1+\sqrt{\frac{8}{3}N-\frac{5}{3}}}{2} \bigg\rfloor\,.$
\end{claim}
\begin{proof}
    Similar to proof of~\Cref{cl:vs4lb}, solving the linear program with upper bound $3i$ in~\Cref{eq:capacity_constraint}, and modifying~\Cref{eq:deg3_max_coverage}.
    We solve for $k$ to satisfy the inequalities
    \begin{align}
        \frac{3}{2}k^2+\frac{3}{2}k+1\leq N< \frac{3}{2}(k+1)^2+\frac{3}{2}(k+1)+1\, ,
    \end{align}
    which leads to the desired result.
\end{proof}

\begin{corollary}
    The lower bound in~\Cref{cl:vs3lb} converges to
        $\frac{\sqrt{\frac{8}{3}N}}{3}\approx 0.54\sqrt{N}$ as $N\to\infty$.
\end{corollary}
As a result of both~\Cref{cl:vs4lb} and~\Cref{cl:vs3lb}, we see that the ASPL lower bound scaling for vertex-symmetric topologies with degree 3 and 4 is $\Theta(\sqrt{N})$.
\subsection{General Regular Topologies}

We minimize ASPL in the general regular case by maximizing the number of nodes at each level.
For $h\geq 1$, the maximum number of nodes reached in exactly $h$-hops would be $\Delta(\Delta-1)^{h-1}$, which grows exponentially with $h$. 
Originally shown in~\cite{cerf1974lower}, the lower bound on ASPL for general regular topologies is restated here.
\begin{theorem}[Corollary 1 from~\cite{cerf1974lower}]
    Let $f(i)=\Delta(\Delta-1)^{i-1}$. The lower bound on ASPL for a graph $G$ is given by
    \begin{dmath}\label{eq:moore_lb}
        A(G)\geq\frac{1}{N-1}\left[\sum_{i=1}^k if(i)+k\left(N -1-\sum_{j=1}^k f(j)\right)\right]\,,
    \end{dmath}
    where k is the largest positive integer such that $N -1-\sum_{j=1}^k f(j)\geq 0\, $ and given by
    $k=\bigg\lfloor\log_{\Delta-1}\left(\frac{(N+\Delta-1)(\Delta-2)}{\Delta}\right)\bigg\rfloor\,.$
\end{theorem}
Recall that in the vertex-symmetric case the scaling of ASPL is $\Theta\left(\sqrt{N}\right)$ for degree 3 and 4. Without the symmetry requirement, the general regular topology ASPL scaling improves to $\Theta\left(\log N\right)$.

A graph that achieves the lower bound in~\Cref{eq:moore_lb} is called a Generalized Moore Graph~\cite{cerf1974lower}.
While Generalized Moore Graphs exist for certain $\Delta$ and $N$, they are not guaranteed to exist for all degrees and sizes and are difficult to find when they do exist~\cite{dalfo2019survey,satotani2018depth}. Therefore in~\Cref{subsec:genregtop_results}, we use heuristics for finding general regular topologies with good ASPL performance.

Next, we analyze ASPL-optimal topologies and lower bound achievability.
\section{Optimal Topology Design}\label{sec:results}
\subsection{Degree 4 Vertex-Symmetric Topologies}
Given a $n_s\times n_o$ constellation, we search for configurations that meet the ASPL lower bound. We first  show that for sufficiently large constellations, the mesh grid will have a strictly larger ASPL than the lower bound.

\begin{claim}\label{cl:grid_bigger_lb}
    For a $n_s\times n_o$ constellation with $n_s+n_o\geq 10$ and $\Delta=4$, the ASPL for the mesh grid topology is strictly greater than the vertex-symmetric lower bound.
\end{claim}
\begin{proof}
    The diameter of the mesh grid is $$D(G_{MG})=\bigg\lfloor\frac{n_s}{2}\bigg\rfloor + \bigg\lfloor\frac{n_o}{2}\bigg\rfloor\,,$$

    whereas the diameter of a vertex-symmetric graph that would meet the lower bound performance would be $$D(G_{LB})=\bigg\lceil\frac{-1+\sqrt{2N-1}}{2} \bigg\rceil\,.$$

    Based on these relations, we get the following chain of inequalities:
    \begin{align*}
        D(G_{LB})&=\bigg\lceil\frac{-1+\sqrt{2N-1}}{2} \bigg\rceil  \stackrel{(*)}{<}\bigg\lceil \frac{n_s}{2}+\frac{n_o}{2}-1 \bigg\rceil\\
        &\leq \bigg\lfloor \frac{n_s}{2} \bigg \rfloor +\bigg\lfloor\frac{n_o}{2}\bigg\rfloor=D(G_{MG})\,.
    \end{align*}
It is straightforward to show that $(*)$ holds when $n_s+n_o\geq 10$. 
Since for $n_s+n_o\geq 10$ the diameter of the mesh grid topology is strictly greater than that of the diameter lower bound, we can deduce that the mesh grid cannot achieve the maximum packing required to minimize the objective in~\Cref{eq:packing_objective}. Hence, the mesh grid has a strictly larger ASPL. 
\end{proof}
Based on~\Cref{cl:grid_bigger_lb}, the standard mesh-grid configuration is not the best vertex-symmetric topology for optimal ASPL. 
Therefore, we explore conditions for which the ASPL lower bound is achieved exactly. If $n_s$ and $n_o$ are co-prime, there exists a construction for which the ASPL lower bound is achieved. 

\begin{claim}\label{cl:mncoprime}
    If $n_s$ and $ n_o$ are co-prime, then there exists a vertex-symmetric topology that achieves the ASPL lower bound.
\end{claim}
\begin{proof}[Proof Sketch]
    The complete proof is in the appendix. We prove by construction. First, we construct a circulant graph of size $N$ that meets the degree 4 ASPL lower bound. Then, we show there exists an isomorphism between the circulant graph and an equivalent vertex-symmetric graph. Being isomorphic, the circulant and vertex-symmetric graphs have the same ASPL.
\end{proof}
\Cref{cl:mncoprime} suggests that the degree-4 ASPL lower bound is indeed tight depending on the constellation parameters. However, requiring $n_s$ and $n_o$ to be co-prime is a restrictive condition, so we investigate alternative constructions. 
We now show that as long as $n_s$ is a multiple of $n_o$ with some conditions on $n_s/n_o$, there exist topologies that achieve the lower bound in ASPL performance, and the ISL link ranges required to construct such a topology are comparable to the link ranges required for a mesh grid. Note that the following results can also hold for $n_o>n_s$ if the axes are switched for all the constructions.
\begin{claim}\label{cl:mcn}
    For $\frac{n_s}{n_o}=\frac{m^2}{2}$ for some positive integer $m$, there exists a vertex-symmetric topology which achieves the lower bound from~\Cref{cl:vs4lb}.
\end{claim}
\begin{proof}[Proof Sketch]
    The complete proof is in the appendix. We prove this by construction. In particular, we prove that the jump set $\left\{[1\ 0]^\top,[\omega\ 1]^\top\right\}$, where $\omega=m-1$, achieves the lower bound.
\end{proof}
\begin{figure}
    \centering
    \includegraphics[width=0.45\textwidth]{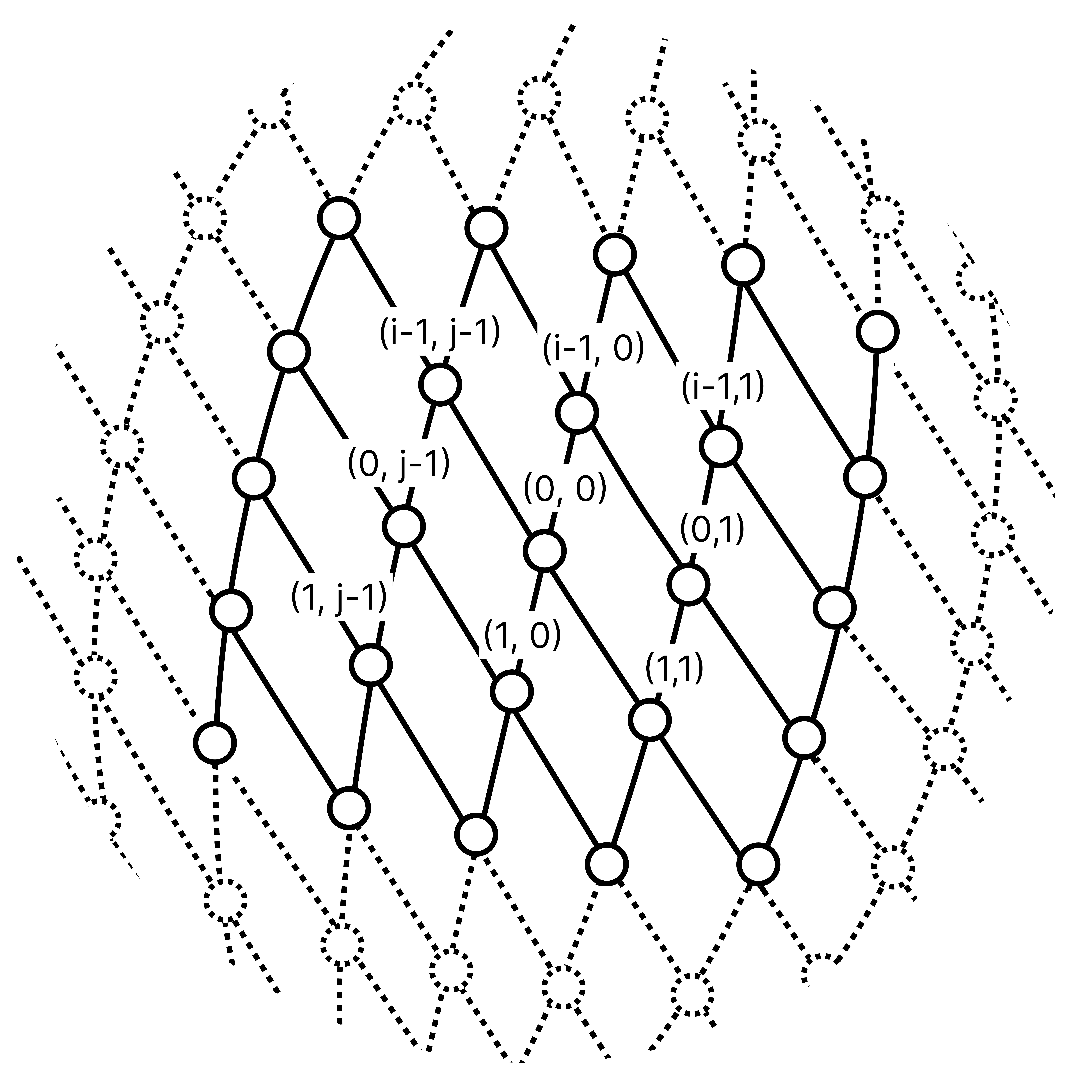}
    \caption{Example topology with fixed intra-orbital links and inter-orbital links with an offset of $1$ (note that for the standard mesh grid the offset is 0). $n_s=i$ and $n_o=j$. When $n_s/n_o=2$, this configuration meets the ASPL lower bound (see~\Cref{cl:mcn}).}
    \label{fig:twisted_torus}
\end{figure}
 
The construction in the proof of~\Cref{cl:vs4lb} is notable for two reasons:
a) the resulting topology is qualitatively similar to the mesh grid but with an offset in the cross link~(see~\Cref{fig:twisted_torus}), and b) the optimal topology still maintains ISL connections between adjacent satellites in the orbital plane. This latter point is important for topology design, as the intra-orbital plane ISL links are more stable~\cite{soret2019inter}. 
Based on this construction, several constellations and topologies exist that achieve the lower bound while maintaining intra-orbital plane ISL connections and cross links to satellites in adjacent orbital planes. A closer inspection of the cross-link jump $[\omega\ 1]^\top$ shows that the value of $\omega$ is square-root proportional to $n_s$:
\begin{align}
    \omega=\sqrt{\frac{2n_s}{n_o}}-1\,.
\end{align}
This square-root proportionality ensures that the required link range to maintain the topology does not increase. That is, as $n_s$ increases, the longitudinal distances between nodes decrease proportionally to $\frac{1}{n_s}$ (see~\Cref{fig:earth_to_torus}) while the vertical offset increases on the order of $\sqrt{n_s}$. Therefore, the overall physical distance the cross-links need to traverse would roughly scale as $\frac{1}{\sqrt{n_s}}$. As a result, the more densely populated each orbital plane is, the less likely the ASPL optimal offset $\omega$ imposes a constraint on the link range.

Based on~\Cref{cl:mcn}, we see that when the ratio $n_s/n_o$ meets the right criterion, we can use the offset-based construction to achieve the ASPL lower bound performance. As $n_s$ is scaled with $n_o$ fixed, this criterion is met infinitely often. Therefore, if $A^*(n_s,n_o)$ denotes the best ASPL possible for a $n_s\times n_o$ vertex-symmetric graph and $A_{LB}(n_s,n_o)$ is the corresponding lower bound,~\Cref{cl:mcn} implies $\liminf_{n_s\to\infty}\frac{A^*\left(n_s,n_o\right)}{A_{LB}(n_s,n_o)}=1$. We wish to investigate how well this square root proportional offset-based construction performs for general $n_s$ and $n_o$ when $n_s>n_o$. In~\Cref{fig:deg4scaling} we plot this topology, choosing the offset to be $\sqrt{\frac{2n_S}{n_o}}-1$ rounded to the nearest integer. For comparison, we find the topology with the best offset -- still keeping intra-orbital links fixed -- via simulation by exhaustive search. The result is shown in~\Cref{fig:deg4scaling} holding $n_o$ constant ($n_o=4$) and scaling $n_s$.

\begin{figure}
    \centering
    \includegraphics[width=0.4\textwidth]{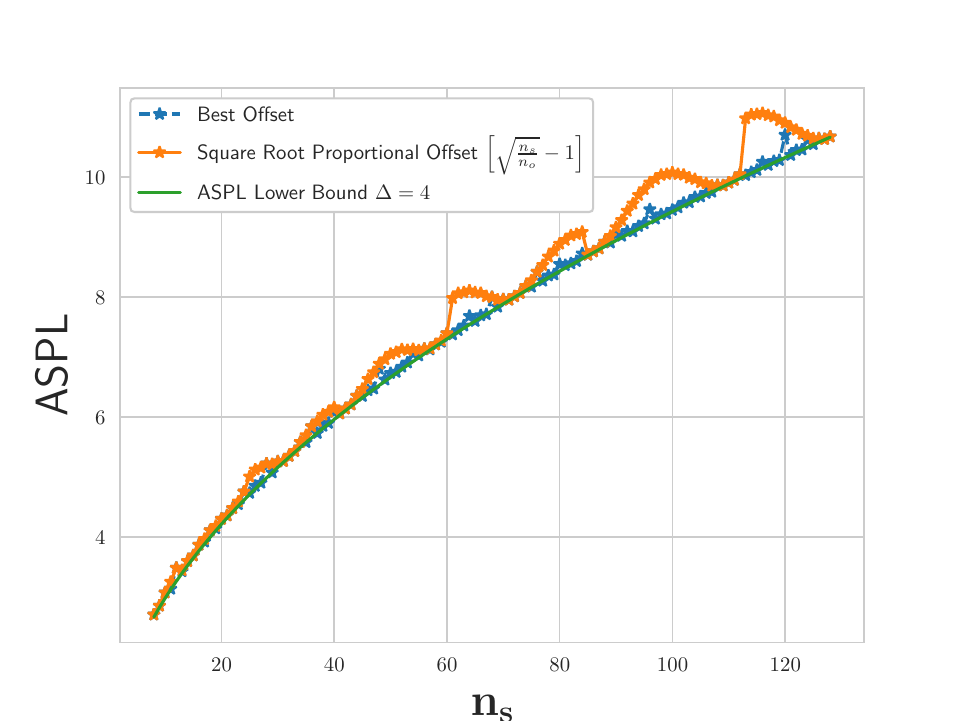}
    \caption{Figure showing the ASPL performance of topologies with fixed $n_o=4$ and $n_s=10, 11,\hdots, 160$.}
    \label{fig:deg4scaling}
\end{figure}
As can be seen in~\Cref{fig:deg4scaling}, the ASPL of the square-root proportional offset-based construction remains close to the lower bound. The optimal-offset topology performance is consistently closer to the lower bound, often achieving the lower bound performance.
This observation leads us to surmise that the best possible ASPL topology as $n_s$ increases and $n_o$ is fixed remains very close to the lower bound. Specifically, we conjecture the following.

\noindent \textbf{Conjecture}: \textit{for fixed $n_o$, }
\begin{align}
    \limsup_{n_s\to\infty}\frac{A^*(n_s,n_o)}{A_{LB}(n_s, n_o)}=1\,.
\end{align}

The conjecture posits that one can always find a vertex-symmetric topology that is close in ASPL performance to the lower bound. 
\subsection{Degree 3 Vertex-Symmetric Topologies}
Next, we investigate ASPL performance of degree-3 topologies. A standard 3-regular topology that is vertex-symmetric is the honeycomb mesh topology, which can be defined with the jump set $\{[1\  0]^\top, [-1\ 0]$, $[0\ 1]^\top\}$. Two jumps are along one axis, and the remaining jump is along the other axis. 

Similar to the degree 4 case, the honeycomb mesh cannot achieve the vertex-symmetric lower bound, since the diameter of the honeycomb mesh is  $\lfloor \frac{n_s}{2}\rfloor + \lfloor \frac{n_o}{2}\rfloor$ as well. We now state a similar result to that of~\Cref{cl:grid_bigger_lb} for the $\Delta=3$ vertex-symmetric case.
\begin{claim}\label{cl:3_grid_bigger_lb}
    For a $n_s\times n_o$ constellation with $n_s+n_o\geq 16$ and $\Delta=3$, the ASPL for the honeycomb topology is strictly greater than the vertex-symmetric lower bound.
\end{claim}
The proof is similar to the proof of~\Cref{cl:grid_bigger_lb} and is omitted for brevity.
An intuitive explanation for why both the honeycomb mesh and the mesh grid cannot achieve the respective lower bounds is that both topologies' diameters scale linearly with $n_s$ and $n_o$. Since the jump definitions are fixed for both topologies, they do not change in response to the change in network size. A topology which achieves the best ASPL performance would meet the diameter lower bound, which scales $\Theta(\sqrt{n_sn_o})$ (see~\Cref{eq:diam_lb}).

While few constellations have been found for which the degree 3 vertex-symmetric topology lower bound holds, we can show that when $n_s/n_o=6$, there exists a topology that achieves exactly the lower bound performance. The claim is presented in the following, and the proof is omitted for brevity.
\begin{claim} \label{cl:d3m6n}
For $\frac{n_s}{n_o}=6$ and $\Delta =3$, there exists a vertex-symmetric topology such that the vertex-symmetric lower bound performance is achieved.
\end{claim}
\begin{proof}[Proof Sketch]
    We prove by construction, using the jump set $\{[1\ 0]^\top,[3\ 0]^\top,[5\ 1]^\top\}$ to demonstrate that the lower bound as described in~\Cref{cl:vs3lb} is achieved. 
\end{proof}

\begin{figure}
    \centering
    \includegraphics[width=0.4\textwidth]{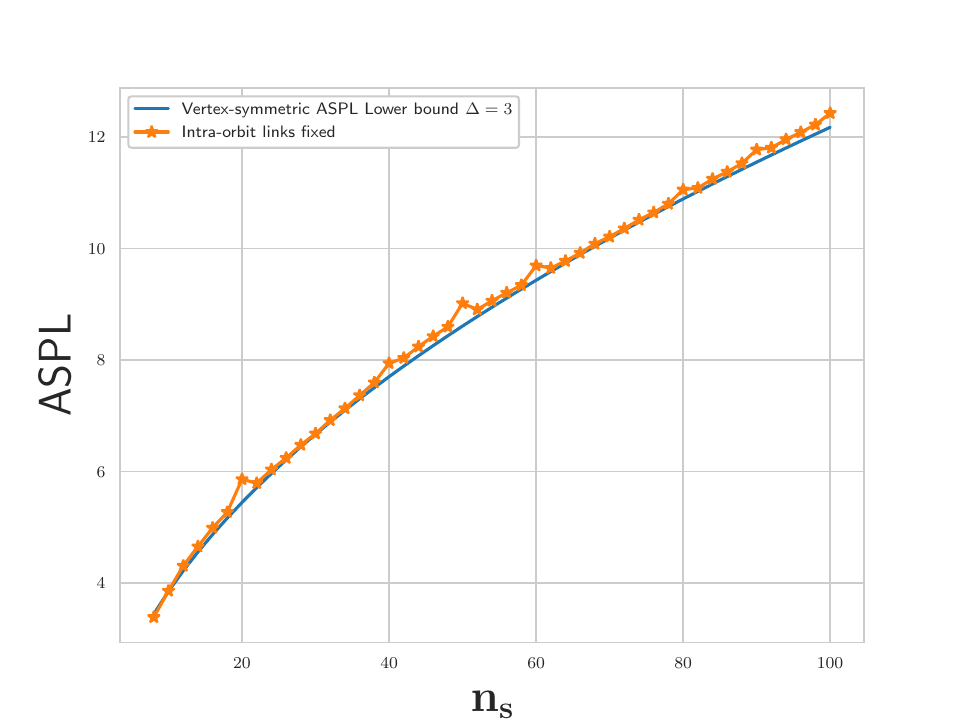}
    \caption{ ASPL performance of vertex-symmetric degree 3 topologies found by fixing intra-orbit links and searching for best inter-orbit link offset.The number of orbits is fixed $n_o=5$ and $n_s$ takes values $8,10,12,\hdots,88$.}
    \label{fig:deg3_scaling}
\end{figure}

    For general $n_s$ and $n_o$, we conjecture that as $n_s$ is scaled and $n_o$ is fixed there exist topologies that are close to the lower bound, as demonstrated in~\Cref{fig:deg3_scaling}.  An exhaustive search for the ASPL-optimal topology demonstrates that even if the lower bound is not exactly achievable for many configurations, there still exist vertex-symmetric topologies that perform close to the lower bound. Thus, we conjecture that $\limsup_{n_s\to\infty}\frac{A^*\left(n_s,n_o\right)}{A_{LB}(n_s,n_o)}=1$ holds for degree 3 as well.

    While designing vertex-symmetric topologies may be favorable for the regular network structure and the reduced routing complexity, vertex-symmetry imposes a rigidity in the design that restricts the ASPL scaling to be $\Theta(\sqrt{N})$ at best. In the following subsection we will analyze the ASPL performance when $\Delta$-regular networks are no longer restricted to being symmetric. The ASPL performance then improves to a logarithmic scaling behavior even when physical constraints such as link range are accounted for.

\subsection{General Regular Topologies}\label{subsec:genregtop_results}
As mentioned before, finding Generalized Moore Graphs is difficult, and for a given $\Delta$ and $N$ there is no guarantee that a Generalized Moore Graph exists.
Thus, one resorts to heuristics to find Generalized Moore Graphs or near-optimal regular graphs~\cite{satotani2018depth,hirayama2022faster}.  We use simulated annealing to find low-ASPL regular topologies and study whether it can scale with the ASPL lower bound. Choosing a good initial graph as a starting point is critical for simulated annealing, so we use a result from random graph theory to inform the choice of starting topology. If a graph is sampled uniformly at random from the space of all $\Delta$-regular graphs of size $N$, then with high probability the diameter of the graph is $\Theta(\log_{\Delta-1}(N))$~\cite{bollobas1998random}. Since the ASPL of a graph is upper-bounded by the diameter of the graph, a topology chosen randomly in such a manner will have a logarithmically scaling ASPL with high probability. We use this random sampling technique to choose our initial graph for the simulated annealing routine, using the algorithm in~\cite{steger1999generating}. Broadly speaking, the routine draws two nodes at random and introduces an edge between them if the two nodes are distinct and do not already have an edge connecting them in the graph. The performance of simulated annealing starting with the random sampling routine is shown in~\Cref{fig:sim_annealing}.


The simulated annealing output with random sampling has favorable scaling behavior with respect to ASPL. However, the search space is over all $\Delta$-regular topologies, which contains a large portion of physically infeasible topologies given the spatial arrangement of satellites in a constellation. Therefore, we introduce a link range constraint $r$ such that if the toroidal distance between any two satellites $u$ and $v$ exceeds $r$, the edge $(u,v)$ would be removed from consideration. We introduce an additional condition that the edge is admitted into the graph only if the two randomly sampled nodes lie within a distance $r$ of each other. We compare the unconstrained sampling routine with simulated annealing against the modified routine using a distance-constrained simulated annealing. The result is shown in~\Cref{fig:dist_constrained}, for which $n_s=n_o$ while $n_s$ is scaled from 4 to 25, and the link range is fixed to $r=0.25$. As can be seen from the figure, when the number of nodes in the network increases, the simulated annealing routine produces low-ASPL topologies that scale in proportion to the general regular topology ASPL lower bound (see~\Cref{eq:moore_lb}). 

\begin{figure}
    \centering
    \includegraphics[width=0.4\textwidth]{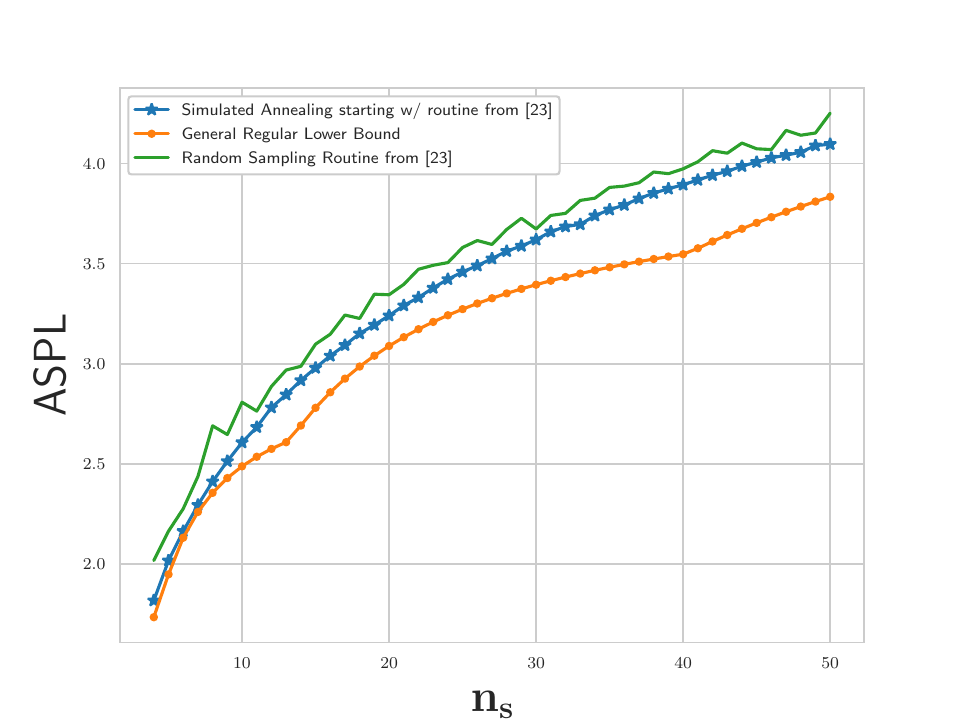}
    \caption{ Simulated Annealing for $\Delta=4$. We fix $n_o=4$}
    \label{fig:sim_annealing}
\end{figure}

The ASPL of topologies produced by the modified routine suggests that with sufficient satellite density, a simple procedure can find topologies with logarithmically scaling ASPL with high probability. To prove this, we devise a procedure called the Partitioning with Random Graph Sampling (PRGS) procedure. We outline PRGS for $\Delta=4$ in the following.
\begin{procsteps}\label[procstepsi]{alg:logs}
    \item For $\Delta=4$, Consider the unit-area square torus with $N$ satellites and a link range constraint $0<r<1$. Let $b$ equal the smallest integer such that $b\geq \frac{\sqrt{2}}{r}$. We divide the torus into $b^2$ square grid cells, each with side length $w=\frac{1}{b}$. Based on the size of each cell, any satellite within a grid cell is within link range of any other satellite in the same grid cell. 
    \item Let $N=cb^2$ for some integer $c\gg 1$, and let $n_s=n_o=\sqrt{N}$. We partition the nodes based on which grid cell they are located in, with nodes on the boundaries assigned to one of the adjacent partitions uniformly at random. 
    \item For any two adjacent cells, find the pair of satellites -- one located in each grid cell -- that are closest in distance to each other and connect them with an edge. For any grid cell, 4 such nodes will have edges that connect with a node in an adjoining grid cell. 
    \item Connect each of these 4 nodes together to form a cycle and treat the four nodes as one super-node with 4 free links. 
    \item Finally, use routine from~\cite{steger1999generating} in each grid cell to connect the vertices plus super node together to create a $4$-regular subgraph. We now have a $4$-regular connected topology of size $N$.
\end{procsteps}

\begin{figure}
    \centering
    \includegraphics[width=0.4\textwidth]{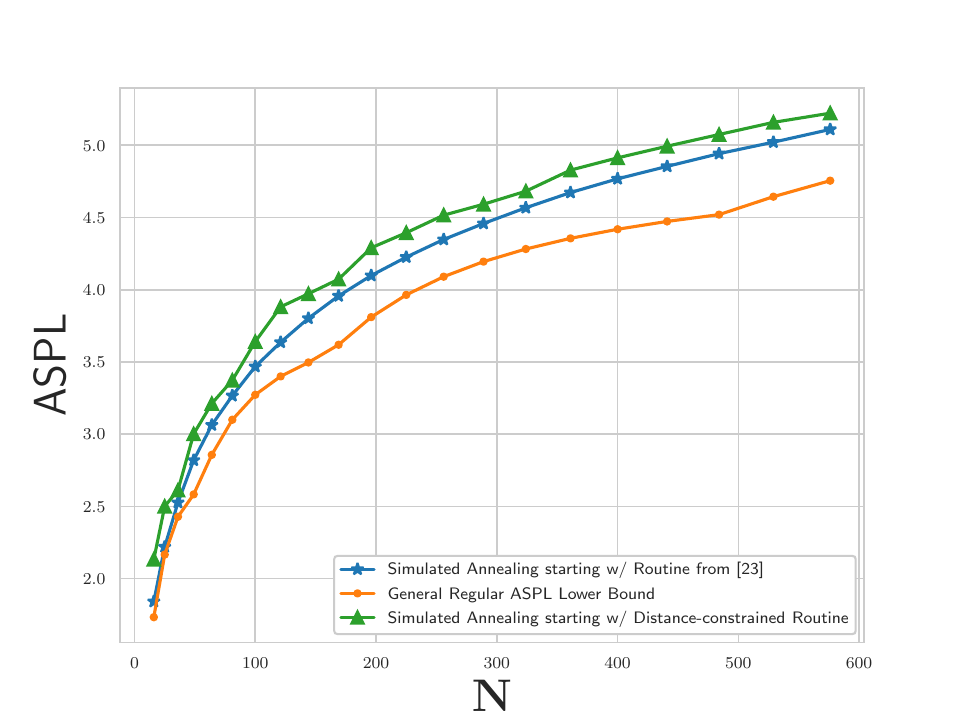}
    \caption{Simulated Annealing with random sampling, both with and without distance constraints. $\Delta=4$ and link range $r=0.25$. We let $n_s=n_o$.}
    \label{fig:dist_constrained}
\end{figure}
 Though described for degree 4, PRGS can be adapted for degree 3 as well. We show that this procedure scales logarithmically in the following.
\begin{claim}\label{cl:part_samp_r}
    The probability PRGS produces a topology with $O(\log N)$ ASPL goes to $1$ as $N\to\infty$.
\end{claim}
\begin{proof}
    The average hop between any two nodes within the same grid cell scales as $\Theta(\log N)$ with probability $\to 1$ by virtue of the result from~\cite{bollobas1998random}. For any two nodes in different grid cells,  it takes $\Theta(\log N)$ hops to reach a boundary node. Then, it takes $\Theta(1)$ hops to reach the grid cell of the destination node since there are a constant number of grid cells. Finally, it takes another $\Theta(\log N)$ hops to reach the destination node from the grid cell boundary.
\end{proof}
Using PRGS, as the size of the constellation increases while $\Delta$ remains constant, it is very likely the uniformly randomly sampled graph has logarithmically scaling diameter, and therefore logarithmically scaling ASPL. As the constellation density increases and the maximum link range remains fixed, the number of satellites within link range increases proportionally to $N$. Therefore, so long as the constellation is sufficiently dense, a graph with low ASPL is physically feasible.

\section{Conclusion}\label{sec:conclusion}

We explore ASPL-optimal ISL topology design and provide analytical lower bounds on performance for two topology design cases: vertex-symmetric and general regular topologies. For the vertex-symmetric case, we demonstrate constellation parameters exist for which the lower bound is attainable for degree 3 and 4. We show that optimal topologies that meet the vertex symmetric lower bound can have potentially physically realizable constructions given the proximity of the ISL links and favorable distance scaling properties. For the general regular topology case, heuristics are used to find low-ASPL topologies and show that the performance scales on the order of the ASPL lower bound for general regular topologies even with link range constraints.

\appendices
\section{Proof of Claim 1}
\begin{proof}
    We find the largest value of $k$ such that $\sum_{i=1}^k x_i<N-1$. Since~\Cref{eq:capacity_constraint} is active for all $i\leq k$ we know that $1+\sum_{i=1}^k x_i=(2k^2+2k+1)$. Since the maximum number of nodes that can be reached in $k+1$ hops is $2(k+1)^2+2(k+1)+1$, we have the following relation:
\begin{align}
    2k^2+2k+1\leq N < 2(k+1)^2+2(k+1)+1\,.
\end{align} 
Solving either the upper or lower quadratic inequality leads to 
    $k=\bigg\lfloor \frac{-1+\sqrt{2N-1}}{2} \bigg\rfloor$.
\end{proof}

\section{Proof of Claim 5}
To prove~\Cref{cl:mncoprime}, we need the following lemma

\begin{lemma}\label{lem:circulant}
Any vertex symmetric topology of a constellation with $n_s$ and $ n_o$ co-prime is isomorphic to a circulant graph.
\end{lemma}
\begin{proof}
     Let $\sigma(\cdot) $ be a function which assigns a node ID via the following recurrence relation:
    \begin{align}\label{eq:bijectivecoprime}
    \sigma(0,0)&=0\\
        \sigma(a,b)&=\sigma(a-1 \pmod {n_s}, b-1\pmod {n_o})+1\  .
    \end{align}
    Since $n_s$ and $n_o$ are co-prime, it is straightforward to verify that~\Cref{eq:bijectivecoprime} is bijective, since the pre-image  of $x\in\{0,\hdots,N-1\}$ is $\left[ x\pmod {n_s}\ x\pmod {n_o} \right]^T$ which is unique and exists for all $x$. To prove the isomorphism, we show that for vertex pair $v$---$u$, where $v,u\in V$ and $v\oplus e=u$, applying the function $\sigma(\cdot)$ results in a vertex pair $\sigma(v)$---$\sigma(u)$ connected by $\sigma(e$). By definition, 
    \begin{align*}
         \sigma(u)&=b\ \text{s.t.}\ \begin{bmatrix}
             b\pmod {n_s}\\ b\pmod {n_o}
         \end{bmatrix}=\begin{bmatrix}
             v_1+e_1\pmod {n_s}\\ v_2+e_2\pmod {n_o}
         \end{bmatrix}\,.
    \end{align*}
Similarly, let $\sigma(v)=a$ and $\sigma(e)=c$.
    Since $u=v\oplus e$, we know that
    \begin{align}
        \begin{bmatrix}
        u_1\\u_2
    \end{bmatrix}=\begin{bmatrix}
        v_1+e_1\pmod {n_s}\\
        v_2+e_2 \pmod {n_o}
    \end{bmatrix}\,,
    \end{align}
    which by definition equals $\begin{bmatrix}
        b\pmod {n_s}\\
        b \pmod {n_o}
    \end{bmatrix}\,.$ Since the function $\sigma(\cdot)$ is injective, we know that $\sigma(v\oplus e)=\sigma(u)=b$, and we have proven the lemma.
\end{proof}
\begin{proof}[Proof of~\Cref{cl:mncoprime}]
    Since~\Cref{lem:circulant} holds, we show there exists a minimum ASPL topology for every circulant graph of any nontrivial size ($N>6$), which is isomorphic to a vertex-symmetric topology. We choose the jump set $\{w,w+1\}$ such that
        $\frac{4w^2+1}{2}\leq N \leq 2(w+1)^2\,.$
    when $N\leq \frac{(2w+1)^2+1}{2}$ it is straightforward to verify that the diameter lower bound for the circulant graph is $w$. Therefore, in no more than $w$ hops, we must span $N$ consecutive integers with the integer linear sum $aw+b(w+1)$ such that 
     \begin{align}
        |x|\leq w(w+1)\label{eq:xless}\\
        x=aw+b(w+1)\label{eq:circ_eq} \\
        |a|+|b|\leq w\,.\label{eq:hopless}
    \end{align}
    Since any two consecutive integers are coprime, $w$ and $w+1$ are coprime. Therefore, given $a,b$ and $x$, all the combinations of integers equal to $x$ can be described with the expression $(a+p(w+1))w+(b-pw)(w+1)$ for an arbitrary integer $p$. We show that in fact $p=0$, and there is only one pair of coefficients that meet all~\Cref{eq:xless,eq:circ_eq,eq:hopless}. Without loss in generality, if $a<0$ and $b>0$, and~\Cref{eq:circ_eq,eq:hopless} hold, then the coefficients $a'=a+p(w+1)$  and $b'=b-pw$ violate~\Cref{eq:hopless} for any non-zero $p$. 
    Since we established that $a$ was negative and $b$ was positive,
        $|a|+|b|=b-a\,.$
    However, $b-a\geq w+1$, which violates~\Cref{eq:hopless}. A similar line of reasoning follows for $b=0$. Thus $p=0$. Since we have established the uniqueness of all integer sums with distinct pairs of $a,b$ which meet the constraints of~\Cref{eq:xless,eq:circ_eq,eq:hopless}, we may quantify the number of distinct pairs, which is $2w^2+2w+1\geq \frac{(2w+1)^2+1}{2}\geq N$. Thus the vertex-symmetric lower bound is met. 

    For $N > \frac{(2w+1)^2+1}{2}$, the diameter lower bound is $w+1$. In addition to the $2w^2+2w+1$ distinct values that can be reached in up to $w$ hops, we can quantify the number of distinct values that are reached in $w+1$ hops, which is $4w+4$. However, it is easy to verify that the only distinct sums are for $a,b$ such that $|aw+b(w+1)|>w(w+1)$ for which there are exactly $2w+2$ distinct coefficient pairs. Therefore, in total $2w^2+2w+1+2w+2>2w^2+4w+2=2(w+1)^2\geq N $ values can be represented, which is sufficient to span all nodes $0,\hdots, N-1$. 
\end{proof}

\section{Proof of Claim 6}
\begin{proof}
   Let $c=n_s/n_o$. Assume $c>2$. It is straightforward to show that the value of $k$ is $\sqrt{\frac{N}{2}}-1=D-1$ where $D$ is the diameter lower bound (see~\Cref{eq:diam_lb}). We first show that for up to $k$ hops the number of nodes covered is maximized.  For $[a\ b]^\top\neq [c\ d]^\top$, we must show that $ae_1\oplus be_2\neq ce_1\oplus de_2$, where $|a|+|b|\leq k$ and same for $c$ and $d$. We start with the case $b=d$. If so, $a=c+pn_s$ for some integer $p$. Since $n_s>2k$, $p$ must be $0$ which contradicts the statement $[a\, b]^\top\neq [c\, d]^\top$. If $a=c$ then $b=d+\frac{pc}{\omega}n_o$. It is straightforward to show that $c$ and $ \omega$ are co-prime, and $p=\ell\omega$, thus $b=d+\ell n_s$. $|\ell n_s|$ exceeds $2D$ for $\ell\neq 0$, thus we have a contradiction. For general $a,b,c,d$, we have the relation $a+p\omega n_o=c$ must hold for some integer $p\in\{-1,1\}$. So long as $c>2$, we can show $\omega n_o> 2D$ and thus the equation cannot hold unless $p=0$. We have now shown that up to $k$ hops the number of nodes is maximized. For nodes exactly $D=k+1$ hops away, there should be $2D-1$ that cannot be covered in less than $D$ hops. We can easily enumerate all $2D-1$ nodes by setting $(a,b)$ to the sequence $(0,D),(1,D-1),(2,D-2),\hdots,(D-1,1)$ as well as $(-1,D-1),(-2,D-2),\hdots,(-D+1,1)$ which adds up to $2D-1$.

   We omit the proof for $c=2$ for brevity, but it follows a simpler line of analysis.
\end{proof}

\bibliography{satref}

\begin{thebibliography}{10}

\bibitem{werner2001topological}
M.~Werner, J.~Frings, F.~Wauquiez, and G.~Maral, ``Topological design, routing
  and capacity dimensioning for isl networks in broadband leo satellite
  systems,'' {\em International Journal of Satellite Communications}, vol.~19,
  no.~6, pp.~499--527, 2001.

\bibitem{handley2018delay}
M.~Handley, ``Delay is not an option: Low latency routing in space,'' in {\em
  Proceedings of the 17th ACM Workshop on Hot Topics in Networks}, pp.~85--91,
  2018.

\bibitem{cao2013per}
J.~Cao, R.~Xia, P.~Yang, C.~Guo, G.~Lu, L.~Yuan, Y.~Zheng, H.~Wu, Y.~Xiong, and
  D.~Maltz, ``Per-packet load-balanced, low-latency routing for clos-based data
  center networks,'' in {\em Proceedings of the ninth ACM conference on
  Emerging networking experiments and technologies}, pp.~49--60, 2013.

\bibitem{zanzi2020laco}
L.~Zanzi, V.~Sciancalepore, A.~Garcia-Saavedra, H.~D. Schotten, and
  X.~Costa-P{\'e}rez, ``Laco: A latency-driven network slicing orchestration in
  beyond-5g networks,'' {\em IEEE Transactions on Wireless Communications},
  vol.~20, no.~1, pp.~667--682, 2020.

\bibitem{wang2008constellation}
J.~Wang, C.~She, and J.~Liu, ``Constellation inference for polar leo satellite
  networks by delay probing,'' {\em European transactions on
  telecommunications}, vol.~19, no.~3, pp.~285--297, 2008.

\bibitem{chen2021analysis}
Q.~Chen, G.~Giambene, L.~Yang, C.~Fan, and X.~Chen, ``Analysis of
  inter-satellite link paths for leo mega-constellation networks,'' {\em IEEE
  Transactions on Vehicular Technology}, vol.~70, no.~3, pp.~2743--2755, 2021.

\bibitem{chen2022leo}
Q.~Chen, L.~Yang, D.~Guo, B.~Ren, J.~Guo, and X.~Chen, ``Leo satellite
  networks: When do all shortest distance paths belong to minimum hop path
  set?,'' {\em IEEE Transactions on Aerospace and Electronic Systems}, vol.~58,
  no.~4, pp.~3730--3734, 2022.

\bibitem{bhattacharjee2023laser}
D.~Bhattacharjee, A.~U. Chaudhry, H.~Yanikomeroglu, P.~Hu, and G.~Lamontagne,
  ``Laser inter-satellite link setup delay: Quantification, impact, and
  tolerable value,'' in {\em 2023 IEEE Wireless Communications and Networking
  Conference (WCNC)}, pp.~1--6, IEEE, 2023.

\bibitem{henderson2000network}
T.~Henderson and R.~Katz, ``Network simulation for leo satellite networks,'' in
  {\em 18th International Communications Satellite Systems Conference and
  Exhibit}, p.~1237, 2000.

\bibitem{soret2019inter}
B.~Soret, I.~Leyva-Mayorga, and P.~Popovski, ``Inter-plane satellite matching
  in dense leo constellations,'' in {\em 2019 IEEE Global Communications
  Conference (GLOBECOM)}, pp.~1--6, IEEE, 2019.

\bibitem{royster2023network}
T.~Royster, J.~Sun, A.~Narula-Tam, and T.~Shake, ``Network performance of pleo
  topologies in a high-inclination walker delta satellite constellation,'' in
  {\em 2023 IEEE Aerospace Conference}, pp.~1--9, IEEE, 2023.

\bibitem{wang2007topological}
J.~Wang, L.~Li, and M.~Zhou, ``Topological dynamics characterization for leo
  satellite networks,'' {\em Computer Networks}, vol.~51, no.~1, pp.~43--53,
  2007.

\bibitem{bhattacherjee_network_2019}
D.~Bhattacherjee and A.~Singla, ``Network topology design at 27,000 km/hour,''
  in {\em Proceedings of the 15th {International} {Conference} on {Emerging}
  {Networking} {Experiments} {And} {Technologies}}, (Orlando Florida),
  pp.~341--354, ACM, Dec. 2019.

\bibitem{sun_capacity_2003}
J.~Sun and E.~Modiano, ``Capacity provisioning and failure recovery for {Low}
  {Earth} {Orbit} satellite constellation,'' {\em International Journal of
  Satellite Communications and Networking}, vol.~21, no.~3, pp.~259--284, 2003.
\newblock \_eprint: https://onlinelibrary.wiley.com/doi/pdf/10.1002/sat.752.

\bibitem{modiano1996efficient}
E.~Modiano and A.~Ephremides, ``Efficient algorithms for performing packet
  broadcasts in a mesh network,'' {\em IEEE/ACM transactions on networking},
  vol.~4, no.~4, pp.~639--648, 1996.

\bibitem{stojmenovic1997honeycomb}
I.~Stojmenovic, ``Honeycomb networks: Topological properties and communication
  algorithms,'' {\em IEEE Transactions on parallel and distributed systems},
  vol.~8, no.~10, pp.~1036--1042, 1997.

\bibitem{boesch_reliable_1985}
F.~Boesch and {Jhing-Fa Wang}, ``Reliable circulant networks with minimum
  transmission delay,'' {\em IEEE Transactions on Circuits and Systems},
  vol.~32, pp.~1286--1291, Dec. 1985.

\bibitem{cerf1974lower}
V.~G. Cerf, D.~D. Cowan, R.~Mullin, and R.~Stanton, ``A lower bound on the
  average shortest path length in regular graphs,'' {\em Networks}, vol.~4,
  no.~4, pp.~335--342, 1974.

\bibitem{dalfo2019survey}
C.~Dalf{\'o}, ``A survey on the missing moore graph,'' {\em Linear Algebra and
  its Applications}, vol.~569, pp.~1--14, 2019.

\bibitem{satotani2018depth}
Y.~Satotani and N.~Takahashi, ``Depth-first search algorithms for finding a
  generalized moore graph,'' in {\em TENCON 2018-2018 IEEE Region 10
  Conference}, pp.~0832--0837, IEEE, 2018.

\bibitem{hirayama2022faster}
T.~Hirayama, T.~Migita, and N.~Takahashi, ``A faster algorithm to search for
  generalized moore graphs,'' in {\em TENCON 2022-2022 IEEE Region 10
  Conference (TENCON)}, pp.~1--6, IEEE, 2022.

\bibitem{bollobas1998random}
B.~Bollob{\'a}s and B.~Bollob{\'a}s, {\em Random graphs}.
\newblock Springer, 1998.

\bibitem{steger1999generating}
A.~Steger and N.~C. Wormald, ``Generating random regular graphs quickly,'' {\em
  Combinatorics, Probability and Computing}, vol.~8, no.~4, pp.~377--396, 1999.

\end{thebibliography}
\bibliographystyle{ieeetr}

\end{document}